\documentclass[11pt]{article}
\usepackage{amsfonts,amssymb,amsmath,latexsym,ae,aecompl}
\usepackage{graphics}
\usepackage{latexsym}
\usepackage{algorithm}
\usepackage{algorithmic}
\usepackage{tikz}
\usepackage{subfig}
\newcommand{\qed}{\hfill$\Box$}
\newenvironment{proof}{\noindent {\bf Proof.}}{\qed}

\newtheorem{theorem}{Theorem}[section]
\newtheorem{lemma}{Lemma}[section]

\begin{document}

\baselineskip 0.2in
\parskip      0.1in
\parindent    0em

\bibliographystyle{plain}

\title{{\bf Deterministic Collision-Free Exploration\\ of Unknown Anonymous Graphs}}

\author{
Subhash Bhagat \footnotemark[1]
\and
Andrzej Pelc  \footnotemark[2] \footnotemark[3]
}

\date{ }
\maketitle
\def\thefootnote{\fnsymbol{footnote}}

\footnotetext[1]{
\noindent
Department of Mathematics, Indian Institute of Technology Jodhpur, India. {\tt sbhagat@iitj.ac.in}}

\footnotetext[2]{
	\noindent
	D\'{e}partement d'informatique, Universit\'{e} du Qu\'{e}bec en Outaouais,
	Gatineau, Qu\'{e}bec,
	Canada. {\tt pelc@uqo.ca}}

\footnotetext[3]{
Research supported in part by NSERC  Discovery Grant 2018-03899  and by the
Research Chair in Distributed Computing of the
Universit\'{e} du Qu\'{e}bec en Outaouais.
}

\begin{abstract} 
	We consider the fundamental task of network exploration. A network is modeled as a simple connected undirected $n$-node graph with unlabeled nodes,
	and all ports at any node of degree $d$ are arbitrarily numbered $0,\dots, d-1$. Each of two identical mobile agents, initially
	situated at distinct nodes, has to visit all nodes and stop. Agents execute the same deterministic algorithm and move in synchronous rounds: in each round an agent can either remain at the same node or move to an adjacent node. Exploration must be collision-free: in every round at most one agent can be at any node. We assume that agents have vision of radius 2: an awake agent situated at a node $v$ can see the subgraph induced by all nodes at distance at most 2 from $v$,  sees all port numbers in this subgraph and the agents located at these nodes. Agents do not know the entire graph but they know an upper bound $n$ on its size. The {\em time} of an exploration is the number of rounds since the wakeup of the later agent to the termination by both agents. We show a collision-free exploration algorithm working in time polynomial in $n$, for arbitrary graphs of size larger than 2. Moreover we show that if agents have only vision of radius 1, then collision-free exploration is impossible, e.g., in any tree of diameter 2.

\vspace*{1cm}


\end{abstract}
 \section{Introduction}
 %%%%%%%%%%%%%%%%%%%%%%%%%%%%%%%%%%%%%%%%%%%%%%%%%%%%%%%%%%%
 
 Exploration of networks by visiting all of their nodes is one of the basic tasks performed by mobile agents in networks.
 Agents may model humans walking along streets of a town with the aim of visiting it, software agents
 collecting data placed at nodes of a computer network, or mobile robots
 collecting samples of air or ground in a contaminated mine whose corridors form links of a network. 
 We consider exploration by two agents, where each node must be visited by each of the agents, and both agents have to eventually stop.
 We impose the additional requirement that exploration must be {\em collision-free}:  in every round at most one agent can be at any node. This requirement, previously considered in \cite{CDGKKP,NMH}, has many applications. If agents are humans circulating in a pharmacy or a supermarket during a pandemic, some areas (such as cash registers or prescription counters) may be too small to preserve the 6-feet requirement of social distancing between customers. When mobile entities model software agents, these entities may need exclusive access to the part of a distributed data base situated at a node. Finally, mobile robots distributing chemicals at nodes of a network may need to avoid being at a small distance from each other.

 The network is modeled as a simple connected
 undirected graph $G=(V,E)$, called {\em graph} in the sequel.
 Nodes are unlabeled, and all ports at any node of degree $d$ are arbitrarily numbered $0,\dots, d-1$.
 This is a standard assumption in many papers on network exploration \cite{AKLLR,CDK,GP,Re}.
 Each of two identical mobile agents, initially
 situated at distinct nodes, has to visit all nodes and stop. Agents execute the same deterministic algorithm and move in synchronous rounds: in each round, an agent can either remain at the same node or move to an adjacent node by taking one of the ports.
 We say that an agent is located at node $v$ in round $t$, if it is at this node after the above action in round $t$.
 When agents cross each other moving simultaneously along an edge in opposite directions, they do not even notice this fact. 
 Agents are woken up in possibly different rounds, chosen by an adversary. 
 Each agent starts the execution of the algorithm in the round following its wakeup.
 In every round, at most one agent can be located at any node. 
 For a non-negative integer $r$, we say that agents have {\em vision of radius} $r$, if in each round $t$,  an awake agent located at a node $v$ in round $t-1$ sees the subgraph $G(v,r)$ induced by all nodes at distance at most $r$ from $v$, sees all port numbers in this subgraph, and sees the other agent located at any of these nodes in round $t-1$. We call this information the {\em input} of the agent in round $t$. 
 We assume that agents have vision of radius 2. In each round $t$, an awake agent located at a node $v$ in round $t-1$, uses the input in round $t$ to determine its action in round $t$: either moving to an adjacent node choosing one of the ports at $v$, or staying idle at $v$. 
 
 An awake agent can compare its inputs in consecutive rounds and deduce a crucial piece of information (often used in our algorithm), namely that the other agent has moved. We will say that agent $a$ already awake in round $t-1$ {\em sees the agent $b$ moving in round $t$}, if a node $u$ is in the input of $a$ in rounds $t-1$ and $t$, and one of the following conditions is satisfied: either $b$ was not at $u$ in round $t-1$ but was at $u$ in round $t$, or $b$ was at $u$ in round $t-1$ but was not at $u$ in round $t$.

 Agents have no {\em a priori} knowledge of the graph but they know an upper bound $n$ on its size. 
 Agents are computationally unbounded and cannot mark the visited nodes. They cannot send any messages. 
 The {\em time} of an exploration is the number of rounds since the wakeup of the later agent to the termination by both agents.

 \subsection{Our results}
 
 Our main result is a collision-free exploration algorithm working in time polynomial in $n$, for arbitrary graphs of size larger than 2 (in the two-node graph, collision-free exploration is impossible, see Section 3). We also show that  the assumption of vision of radius 2 cannot be weakened: we prove that collision-free exploration with vision of radius 1  is impossible in any tree of diameter 2.
 
 An important building block  of our algorithm is a procedure called $EXP(n)$, that, given any positive integer $n$, allows an agent to visit all nodes of any graph of size at most $n$,
 starting from any node of this graph, using $R(n)$ edge traversals, where $R$ is some polynomial. We show that our collision-free exploration algorithm works in time $O(R(n)\log n)$, and thus has only logarithmic slowdown with respect to the time of exploration by a single agent, without any constraints. It is important to note that, while we use the well-known Reingold's \cite{Re} result as the base of $EXP(n)$, any exploration procedure working in arbitrary anonymous graphs of size at most $n$ would do, and no changes would be necessary. (We could use, for example, the faster but non-constructive exploration algorithm from \cite{AKLLR}). Hence, our result can be viewed as a generic one: given any exploration procedure for one agent, working in arbitrary anonymous graphs of size at most $n$, we obtain a collision-free exploration with only logarithmic slowdown. Of course, if the collision-free restriction were removed, we could use any of the existing exploration procedures for anonymous graphs, and each of the agents would accomplish exploration independently.
 
 Collision-free exploration in anonymous graphs, studied in this paper, should be compared to the scenario considered in \cite{CDGKKP}, where collision-free exploration was studied in {labeled} graphs. As opposed to our assumption that nodes of the graph are unlabeled,  the authors of \cite{CDGKKP} assumed that all nodes of the graph have distinct identifiers and that each agent can see the labels of its current node and of all adjacent nodes. Under this assumption, they designed a cost-efficient collision-free exploration algorithm in arbitrary graphs by teams of agents of arbitrary size. It is important to stress the major difference between the labeled and the anonymous scenario, in the context of collision-free exploration. While the authors of \cite{CDGKKP} where able to precompute collision-free routes of all agents, using distinct labels of nodes, no such technique is possible in our anonymous case. In fact, the main difficulty that we face in the anonymous scenario is breaking symmetry between agents when they are near each other (at distance at most 2), and when there is the danger of colliding at a common neighbor. The possibility of precomputing collision-free routes in the labeled scenario avoids this major difficulty altogether.

 \subsection{Related work}
 
 The problem of
 exploration of an unknown environment by mobile agents
 has been extensively studied in the literature for many decades (cf. the survey \cite{RKSI}).
 Two types of environments have been considered:
 either geometric terrains in the plane, e.g., an
 unknown terrain with convex obstacles \cite{BRS},
 or a room with polygonal \cite{DKP} or rectangular \cite{BBFY} obstacles, or
 as we do in this paper, networks modeled
 as graphs, assuming that the agents may only move along the edges. The graph
 model can be further specified in two different ways:
 either the graph is directed, in which case the agent can move only from
 tail to head of a directed edge  \cite{AH,BFRSV,BS}, or the graph is undirected (as we assume)
 and the agent can traverse edges in both directions  \cite{ABRS,BRS2,DKK,PaPe}.
 The efficiency measure adopted in most papers dealing with exploration of graphs is the cost
 of completing this task, measured by the number of edge traversals by the agent, or by all agents.
 Some authors impose further restrictions on the moves of the agent.
 It is assumed that the agent has either a restricted tank \cite{ABRS,BRS2},
 and thus has to periodically return to the base for refueling, or that it is attached to the
 base by a rope or a cable of restricted length \cite{DKK}.
 
 An important direction of research concerns
 exploration of anonymous graphs.
 In this case it is impossible
 to explore arbitrary graphs and stop after exploration, if no marking of nodes is allowed, and if nothing is known about the graph.
 Hence some authors \cite{BFRSV,BS}
 allow {\em pebbles} which can be dropped on nodes to recognize already visited ones, and
 then removed and dropped on other nodes. A more restrictive scenario assumes
 a stationary token that is fixed at the starting node of the agent \cite{CDK,PeTi}.
 Exploring
 anonymous graphs without the possibility of marking nodes (and thus possibly without stopping)
 is investigated, e.g., in \cite{DFKP,FI}.
 The authors
 concentrate attention not on the cost of exploration but on the minimum amount of memory sufficient
 to carry out this task. In the absence of marking nodes,  in order to guarantee stopping after exploration, some knowledge about the graph is required,
 e.g., an upper bound on its size \cite{CDK,Re}.
 %Exploration of anonymous graphs is also considered in \cite{DM,DJMW,DW}.
 Graph exploration using the paradigm of {\em algorithms with advice} was investigated in \cite{DGMPS,DKM,FIP2,GP}. In this research, the main issue is to find the minimum amount of information that the agent must have in order to perform exploration in an efficient way.
 
 \section{Preliminaries}

 In our algorithm we will use the following procedure, called $EXP(n)$, based on universal exploration sequences (UXS) \cite{Ko}, and derived from the  result of Reingold \cite{Re}. Given any positive integer $n$, it allows the agent to visit all nodes of any graph of size at most $n$,
 starting from any node of this graph, using $R(n)$ edge traversals, where $R$ is some polynomial. The sequence of steps of an agent starting at any node of a graph of size at most $n$ and applying procedure $EXP(n)$ will be called an {\em exploration walk}.
 
 A UXS is an infinite sequence $x_1,x_2,\dots$ of non-negative integers. Given this sequence, whose effective construction follows from  \cite{Ko,Re}, the procedure $EXP(n)$ can be described as follows.
 In step 1, the agent leaves the starting node by port 0. For $i \geq 1$, the agent that entered the current node
 of degree $d$ by some port $p$ in step $i$,
 computes the port $q$ by which it has to exit in step $i+1$ as follows: $q=(p+x_i)\mod d$.
 The result of Reingold implies that if an agent starts at any node $v$ of an arbitrary graph with at most $n$ nodes,  and applies procedure $EXP(n)$, then it will visit all nodes of the graph after $R(n)$ steps.
 
 \section{The algorithm}
 
 We start with the observation that  collision-free exploration is impossible in the two-node graph.  Indeed, suppose that there is a collision-free algorithm for this graph. This algorithm would have to instruct an agent to move $x$ rounds after waking up, for some $x\geq 0$. The adversary can wake up one agent $x+1$ rounds before the other, thus making it move when the other agent is still at the other node, creating a collision. It follows that we can assume from now on that the graph has more than two nodes.
 
 The high-level idea of the algorithm is the following. The agent executes the exploration walk, using procedure $EXP(n)$, until it sees the other agent at some point or until it finishes the walk.  Due to the lack of node labels, such exploration walks cannot be precomputed, and getting at distance at most 2 from the other agent, called an {\em approach}, cannot be avoided in general. An approach is potentially dangerous: if agents blindly followed their exploration walks, they  could get to the same common neighbor in the next round, creating a collision. Even if there is no common neighbor and the agents are at distance 1 in an approach, following the exploration walk is not safe either. The other agent could be not yet woken up, so if the exploration walk prescribes moving to its current location, this could again result in a collision.  The crux of the algorithm is the behavior of the agents during an approach.
 
 At an approach, both agents execute the procedure {\tt Approach resolution} whose aim is to make {\em progress}. Progress is defined as at least one agent making the next step in its exploration walk. The procedure is designed in such a way that at each approach either both agents make the next step in their respective exploration walks, or if at the previous approach only agent $A$ made the next step, then at the current approach agent $B$ will make the next step and agent $A$ will end up at the same node at which it was when this approach occurred. This guarantees that eventually both agents will execute their entire exploration walks without collisions.
 
 We now proceed with the description of procedure {\tt Approach resolution}. The procedure consists of two subroutines: {\tt Synchronization} and {\tt Progress}. The aim of {\tt Synchronization} is to  bring both agents to the same position they occupied at the beginning of this subroutine, and to make sure that they are both woken up and agree on the same round called {\em red}. {\tt Synchronization}  is executed only at the first approach, since at all subsequent approaches  both agents remember that the other agent is awake and, since they see each other in the same round, they can consider this round to be red. Subroutine  {\tt Progress} begins in the round following the red round of an approach and its aim is to make progress, as described above. Upon completion of {\tt Progress}, both agents may stop seeing each other, in which case they continue executing their respective exploration walks until the next approach, or they can still see each other in which case this is considered a new approach, and they proceed to the next execution of {\tt Progress}, as {\tt Synchronization} is not needed anymore.
 
 % We will use the phrase `` an agent $a$ sees the other agent moving in round $i$. This means that agent $a$ sees that the other agent is not at the same node in rounds $i-1$ and $i$.

 \textbf{ Subroutine {\tt Synchronization}}\\
 We now give a detailed description of subroutine {\tt Synchronization}. It starts in the round when at least one of the agents sees the other for the first time. Let $X$ be the set consisting of the two agents. Consider an agent $a\in X$. We denote by $\hat a$ the other agent in $X$. Let agents $a$ and $\hat a$  be located at nodes $u$ and $v$ respectively at the beginning of the first {\it approach}. We refer to $u$ and $v$ as the {\it initial synch nodes} of $a$ and $\hat a$ respectively. Let $B$ be the set of common neighbors of $u$ and $v$. Let $D=V\backslash (B \cup \{u,v\})$. Let $A_u$ and $A_v$  be the subsets  of nodes in $D$ which are adjacent to $u$ and to $v$ respectively. We call these nodes the {\em private} neighbors of $u$ and of $v$ respectively. The back and forth movement of an agent between two fixed nodes is referred to as the {\it oscillating movement}. {\color{black} During the execution of our algorithm, an agent performs three types of {\it waits}, namely, {\it the initial wait, the compulsory wait,} and {\it the terminal wait}. The {\it  initial wait} and the {\it compulsory wait} are performed by an agent before it makes its first move after waking up. Whenever an agent wakes up, it waits for one round. This waiting is referred to as the {\it initial wait} of the agent.  The {\it initial wait} is essential to avoid collisions. The {\it compulsory wait} depends on the situation. There are two reasons for the {\it compulsory wait}: either the agent waits to see the other agent moving for the first time or it waits for the other agent to move to a particular node. This type of wait is also essential to avoid collisions. Finally, the {\it terminal wait} is defined as follows: an agent $a\in X$  waits in round $t+1$ if in round $t$ agent $a$ was moving and finds the other agent $\hat a$ moving for the first time, and in round $t$, agent $a$ is at  its initial synch node. This wait is referred to as the {\it terminal wait} of agent $a$. The {\it terminal wait} is for the confirmation that both agents are at their respective initial synch nodes at the start of the {\it red} round. During the synchronization, the round when for the first time both agents move simultaneously is referred as the {\it event round}. }
 
 %Suppose that agent $a\in X$ wakes up in round $t$. It makes its initial wait in round $t+1$. Now there are two possibilities: in round $t+1$, the other agent $\hat a$ either moves or does not move. The actions of agent $a$ in round $t+2$ and onward depend on this and are as follows:
 
 {\color{black} Synchronization starts in the round when at least one of the agents sees the other agent for the first time. Call this round $\hat t$. Then there are the following two  possibilities,
 	\begin{itemize}
 		\item {\bf Case A:} Both agents see each other moving in round $\hat t$ (figure \ref{fig-01}).
 		\item {\bf Case B:} At least one of the agents does not see the other agent moving in round $\hat t$.
 		
 	\end{itemize} 
 	Now we describe the strategies followed by the agents in both of the above cases separately.
 	
 	{\bf Case A:} In this case, round $\hat t$ is the {\it event round}. Both agents do not move in round $\hat t+1$ and consider round $\hat t+2$ as the red round. We make agents wait in round $\hat t+1$ just for convenience to prove the correctness of our algorithm.

 	\begin{figure}[h]
 		\vspace*{-.294in}
 		\centering
 		\hspace*{-.75in} \includegraphics[scale = .60]{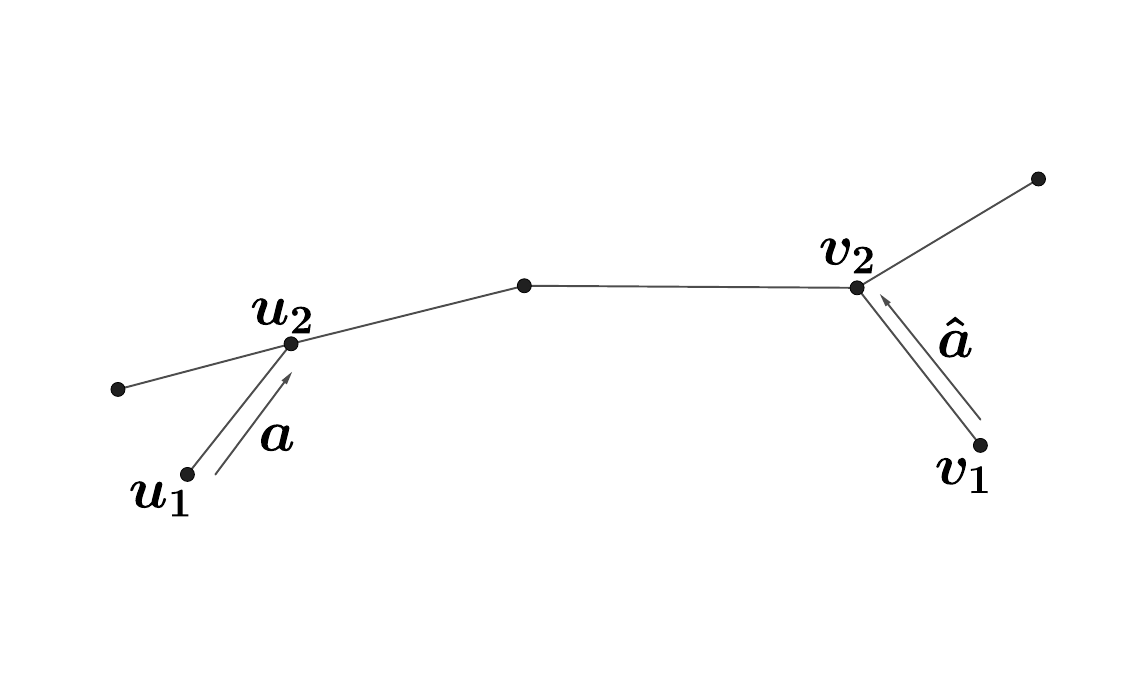}
 		\vspace*{-.5in}
 		\caption{An illustration for the Case A: In round $\hat t-1$, agents $a$ and $\hat a$ were at nodes $u_1$ and $v_1$ respectively. During the round $\hat t$, both agents $a$ and $\hat a$ move concurrently to the nodes $u_2$ and $v_2$ respectively.}
 		\label{fig-01}
 	\end{figure} 	
 	
 	{\bf Case B:}  In this case at least one of the agents wakes up in some round $t\ge \hat t$. Suppose that agent $a\in X$ wakes up in round $t$. It makes its initial wait in round $t+1$. Now there are two possibilities: in round $t+1$, the other agent $\hat a$ either moves or does not move. The actions of agent $a$ in round $t+2$ and onward depend on this and are as follows:}
 
 \begin{itemize}
 	\item {\bf Case B1.  Agent $\hat a$ does not move in round $t+1$:} In this case, agent $a$ infers that the  initial synch node of agent $\hat a$ is the node $v$ at which $\hat a$ stays in round $t+1$. The actions of agent $a$ depend on the sub-graph induced by nodes adjacent to $u$ and to $v$. 
 	\begin{itemize}
 		\item {\bf Subcase 1. Nodes $u$ and $v$ do not have any common neighbor:}
 		In this case, $u$ and $v$ are adjacent.
 		\begin{itemize}
 			\item {\bf Subcase 1.1. $u$  is not a leaf node:} Agent $a$ selects the private neighbor $n_u\in A_u$ which corresponds to the lowest port number at node $u$ and starts the {\it oscillating movement} between its  initial synch node $u$ and the selected private neighbor $n_u$ (figure \ref{fig-02}(A)). It keeps oscillating until it sees the other agent moving for the first time. When $a$ sees $\hat a$ moving for the first time, it stops its oscillating movement. Let $t^*$ be the round when agent $a$ sees agent $\hat a$ moving for the first time (the {\it event round}). If agent $a$ is at its initial synch node $u$ in round $t^*$, then it does not move in round $t^*+1$ (the {\it terminal wait} for agent $a$). Otherwise, agent $a$ moves back to its initial synch node $u$ in round $t^*+1$. Agent $a$ decides on the round $t^*+2$ as the {\it red round}. The  terminal wait of agent $a$ helps the agent $\hat a$ to confirm the initial synch node of $a$. This also ensures that agent $\hat a$ gets back to its initial synch node before the start of the {\it red round}. 
 			
 			\begin{figure}[h]
 				\vspace*{-.1294in}
 				\centering
 				\hspace*{-.75in} \includegraphics[scale = .45]{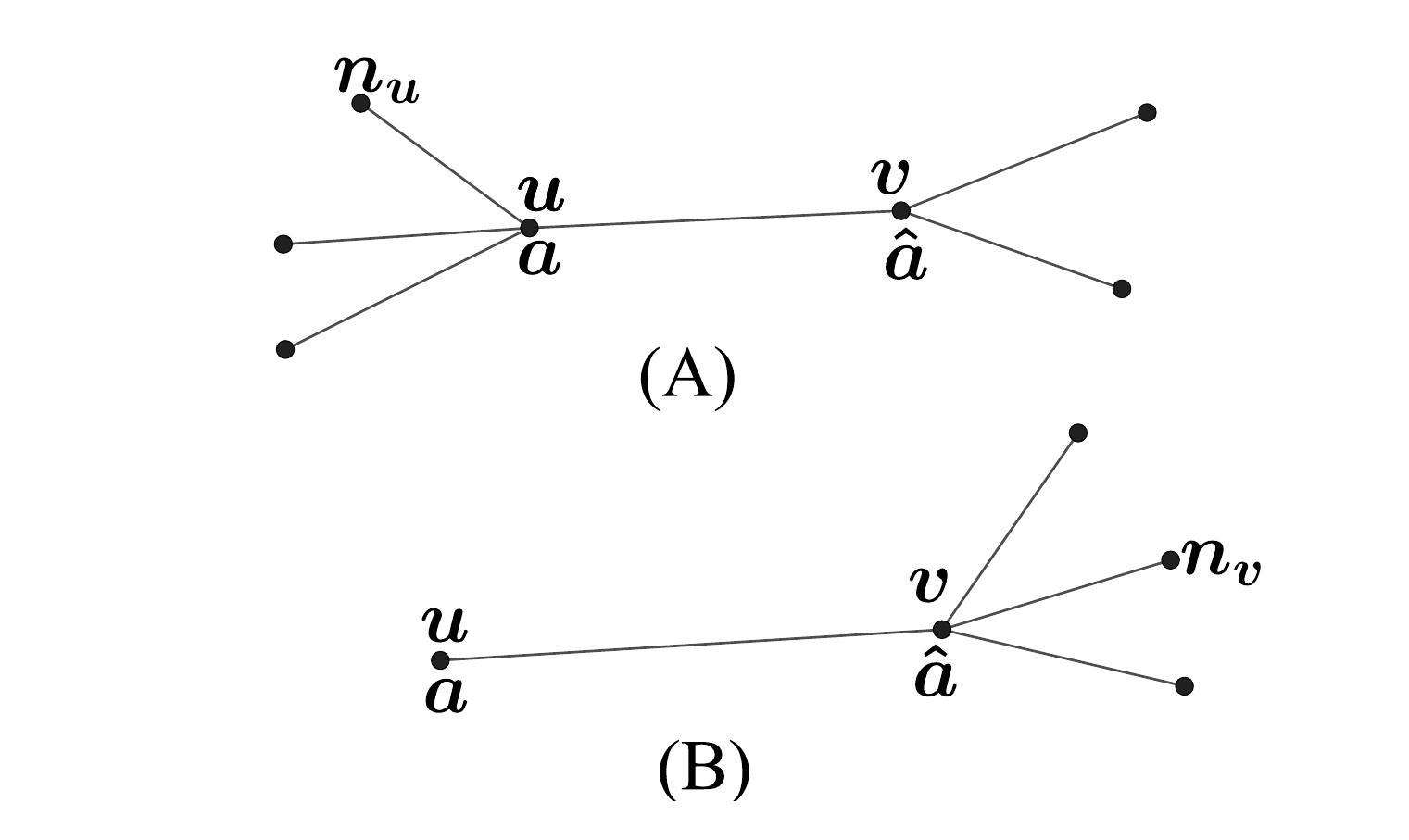}
 				\vspace*{-.1in}
 				\caption{An illustration for: (A) Subcase 1.1 when node $u$ is not a leaf node (B) Subcase 1.2 when node $u$ is a leaf node}
 				\label{fig-02}
 			\end{figure} 	
 			
 			\item {\bf Subcase 1.2. $u$ is a leaf node:} The node $u$ does not have any private neighbor in this case but $v$ does (as there are more than two nodes in the graph). The only node where agent $a$ can move from node $u$ is $v$ (figure \ref{fig-02}(B)). Upon wake-up, agent $a$  waits until it sees the other agent moving (a {\it compulsory wait}). Let $t^*$ be the round when agent $a$ sees agent $\hat a$ moving for the first time. Agent $\hat a$ moves from its initial synch node $v$ to one of its private neighbors $n_v$. In round $t^*$, agent $\hat a$ is at node $n_v$. In round $t^*+1$, agent $a$ does not move (it is a {\it compulsory wait} to avoid collision with agent $\hat a$). In round $t^*+1$, agent $a$ finds agent $\hat a$ at node $v$. In round $t^*+2$ (the event round), agent $a$ moves to node $v$ (while $\hat a$ moves to $n_v$)  and in round $t^*+3$, $a$ moves back to its initial synch node $u$ (while $\hat a$ moves back to $v$). Agent $a$ decides on round $t^*+4$ as the {\it red round}. 
 			
 			%   Now consider the case, when $u$ is not a not leaf node.  In this case, $v$ is a leaf node. Agent $a$ selects the private neighbor $n_u\in A_u$ corresponding to the lowest port number at the agent and starts the {\it oscillating movement} between $u$ and $N_u$.  Let $t^*$ be the round when agent $a$ finds agent $\hat a$ moving for the first time. Agent $a$ knows the initial synch node of agent $\hat a$. At the end of round $t^*$, agent $\hat a$ is at $u$ and agent $a$ is at $n_u$. In round $t^*+1$, agent $a$ moves back to node $u$ and agent $\hat a$ moves back to node $v$. Agent $a$ decides on $t^*+2$ as the {\it red} round. 
 			
 			%Note that agent $a_1$ moves exactly once after it finds $a_2$ moving (it moves back to $u$ while $a_2$ moves back to $v$). This will help the agents to distinguish between this case and a similar kind of situation discussed later in Case 2.2. 
 			
 		\end{itemize}
 		
 		\item {\bf Subcase 2. Nodes $u$ and $v$ have at least one common neighbor:}
 		We define two subsets of nodes $B_u,B_v\subset B$ in the following way: a node $w\in B$ is in $B_u$ if the port number at $w$ corresponding to the edge $\{w,u\}$ is larger than the port number at $w$ corresponding to the edge $\{w,v\}$ (figure \ref{fig-03a}). The set $B_v$ is defined analogously.
 		\begin{itemize}
 			\item {\bf Subcase 2.1. $B_u$ is not empty:} In this case, agent $a$ moves in round $t+2$. Agent $a$  chooses the node $ n_u\in B_u$ corresponding to the lowest port number at $u$, and starts its  oscillating movement between  nodes $u$ and the selected neighbor $n_u$. It continues its movement until it sees the other agent moving. Suppose $t^*$ is the round when agent $a$ sees $\hat a$  moving for the first time. If in round $t^*$, agent $a$ is not at $u$ (its initial synch node), it moves back to $u$ in round $t^*+1$. Otherwise, it waits in round $t^*+1$ (the terminal wait). Agent $a$ decides on $t^*+2$ as the red round. 
 			
 			\begin{figure}[h]
 				\vspace*{-.194in}
 				\centering
 				\hspace*{-.75in} \includegraphics[scale = .45]{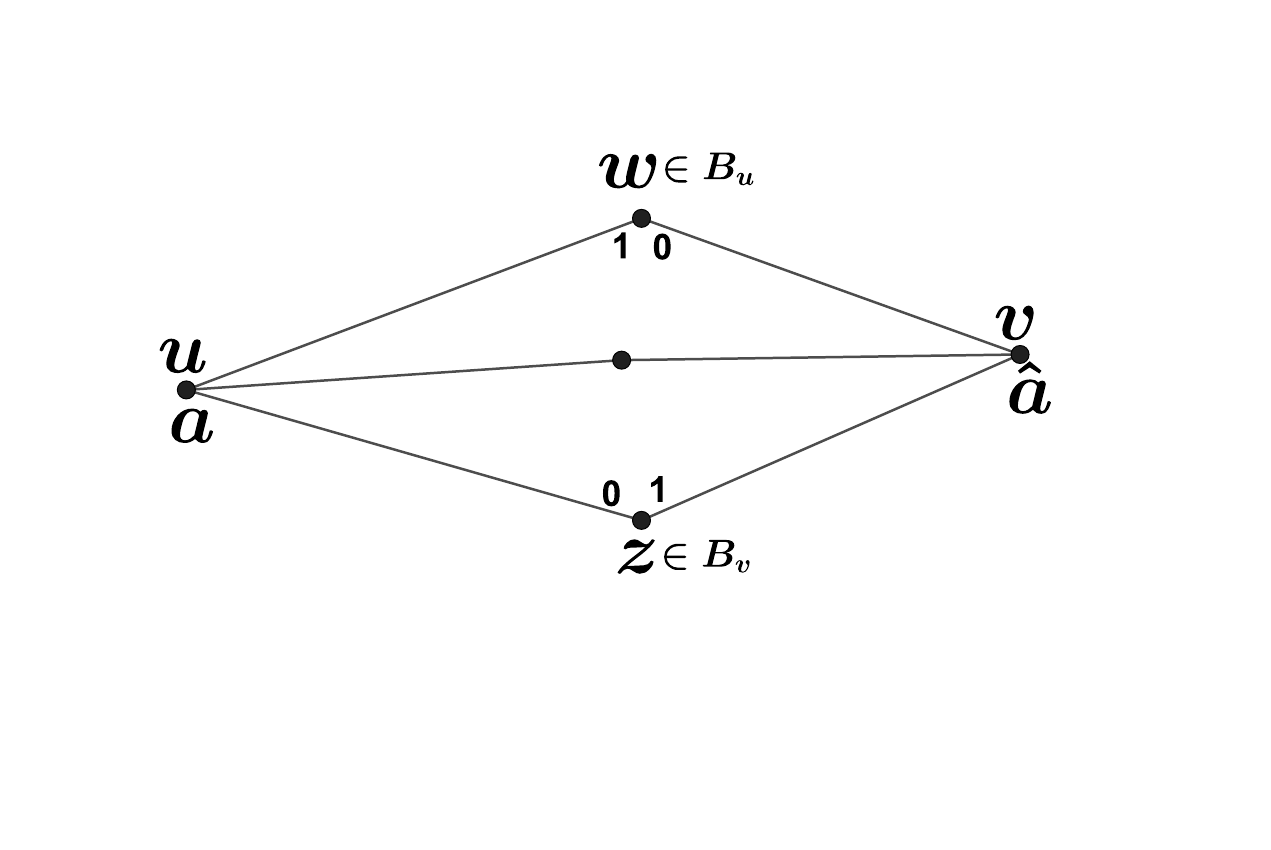}
 				\vspace*{-.6in}
 				\caption{An illustration for Subcase 2 when  nodes $u$ and $v$ have at least one common neighbor}
 				\label{fig-03a}
 			\end{figure}
 			\item {\bf Subcase 2.2. $B_u$ is empty:}  In this case, agent $a$ waits until it sees agent $\hat a$ moving.  Let $t^*$ be the round when agent $a$ sees agent $\hat a$ moving for the first time. In round $t^*$, agent $\hat a$ is at node $n_v$. In round $t^*+1$ (the event round), agent $a$ moves to $n_v$ (while $\hat a$ moves back to $v$). In round $t^*+2$, agent $a$ moves back to node $u$ (while $\hat a$ waits at $v$). Agent $a$ decides on round $t^*+3$ as the {\it red round}.

 		\end{itemize}
 	\end{itemize}
 	\begin{figure}[h]
 		\vspace*{-.23in}
 		\centering
 		\hspace*{-.15in} \includegraphics[scale = .45]{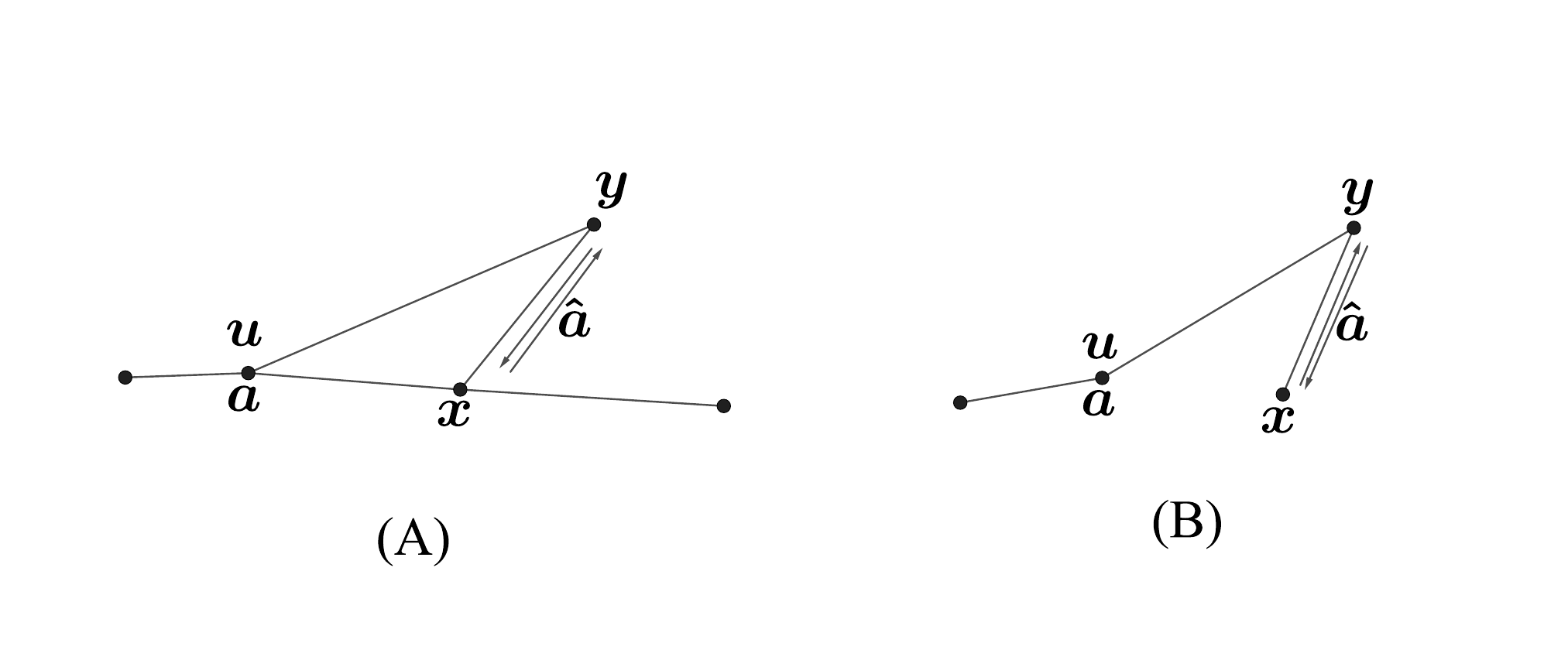}
 		\vspace*{-.6in}
 		\caption{An illustration for Case B2 when agent $\hat a$ moves in round $t+1$: (A) when both nodes $x$ and $y$ are adjacent to node $u$ (B) when exactly one of nodes $x$ and $y$ is adjacent to the node $u$.}
 		\label{fig-03}
 	\end{figure}
 	
 	\item {\bf Case B2. Agent $\hat a$ moves in round $t+1$:}
 	In this case, agent $a$ is not sure which is the initial synch node of agent $\hat a$. Suppose that agent $\hat a$ moves from node $x$ to $y$ in  round $t+1$. One of $x$ and $y$ is the initial synch node of  agent $\hat a$. At least one of $x$ and $y$ is adjacent to node $u$. 
 	First consider the case, when both nodes $x$ and $y$ are adjacent to $u$ (figure \ref{fig-03} (A)). In round $t+2$ (the event round), agent $a$ moves to node $y$ (while $\hat a$ moves back to $x$). In round $t+3$, agent $a$ moves back to its initial synch node $u$. Agent $a$ decides on round $t+4$ as the {\it red round}.
 	
 	Next consider the case, when exactly one of nodes $x$ and $y$ is adjacent to $u$ (figure \ref{fig-03} (B)). If $x$ is not adjacent to $u$ (but $y$ is adjacent to $u$), agent $a$ moves to node $y$ in round $t+2$ (the event round) and moves back to node $u$ in round $t+3$. In this case, agent $a$ decides on round $t+4$ as the {\it red round}. If $x$ is  adjacent to $u$ (but $y$ is not adjacent to $u$), agent $a$ does not move in round $t+2$ (a compulsory wait), while $\hat a$ moves back to $x$. In round $t+3$ (the event round), agent $a$ moves to node $x$ (while $\hat a$ moves to $y$)  and in round $t+4$, $a$ moves back to node $u$. Agent $a$ decides on round $t+5$, as the {\it red round}.
 \end{itemize}

 \begin{lemma}
 	Upon completion of subroutine {\tt Synchronization} both agents establish the same round as red, in this round they know that they are awake and they are at the positions they occupied at the start of the subroutine.
 \end{lemma}
 
 \begin{proof}
 	{\color{black}
 		We show that in all possible scenarios, agents agree on round {\it event round}+2 as the {\it red round}. Suppose that at least one of the agents sees the other agent for the first time in round $\hat t$. We prove the lemma by considering different scenarios separately.
 		
 		{\bf Case A:} {\it Both agents see each other moving in round $\hat t$.}\\ 
 		In this case, round $\hat t$ is the {\it event round}. Agents do not move in round $\hat t+1$ and decide on round $\hat t+2$ as the {\it red round}. Since they do not move in round $\hat t+1$, they remain at the same nodes that they occupy at the start of the subroutine. Hence, the lemma is true in this case.

 		{\bf Case B:} {\it At least one of the agents does not see the other agent moving in round $\hat t$.}\\
 		In this case, at least one of the agents wakes up in round $t\ge \hat t$. We have the following sub-cases in this case: }
 	\begin{itemize}
 		\item {\bf Case 1:} {\it Both agents wake up in the same round.} Suppose that both agents wake up in the same round $t$. The round $t+1$ is the initial wait round for both agents. Since none of them moves in round $t+1$, in this round both agents know the initial synch node of each other. This satisfies condition (iii) above. In round $t+2$, at least one of them moves.
 		\begin{itemize}
 			\item {\bf Case 1.1: Nodes $u$ and $v$ do not have any common neighbor.} We consider two sub-cases.
 			\begin{itemize}
 				\item {\bf Case 1.1.1:  None of $u$ and $v$ is a leaf node.} In this case, both agents are allowed to move simultaneously and they perform {\it oscillating movements} to their respective private neighbors. Since $A_u\cap A_v=\emptyset$, during their movements, agents do not collide with each other. In round $t+2$, for the first time, agents see each other moving. Thus, in this case, round $t+2$ is the {\it event round}. In round $t+2$, agent $a$ is at node $n_u\in A_u$ and agent $\hat a$ is at a node $n_v\in A_v$. Since none of the agents moved in round $t+1$, each of them knows that this is the first move of the other agent in this subroutine, and in round $t+2$ none of them is at its initial synch node. In round $t+3$, both of them move back to their respective initial synch nodes. In round $t+3$, each of the agents knows that the other agent is active and is back at its initial synch node. They both agree on round $t+4$ ({\it event round}+2) as the {\it red round}.

 				\item {\bf Case 1.1.2: One of $u$ and $v$ is a leaf node.} Without loss of generality, suppose that $v$ is a leaf node. Agent $a$ chooses a neighbor $n_u\in A_u$ and starts its  oscillating movement in round $t+2$. Since $v$ is a leaf, the only node where agent $\hat a$ can move, is node $u$. In round $t+1$, agent  $\hat a$ does not know whether agent $a$ is active or not. Thus, agent $\hat a$ does not move in round $t+2$ (a  compulsory wait). In round $t+2$, agent $a$ moves for the first time and it moves to node $n_u$. In round $t+2$, agent $\hat a$ knows that agent $a$ is active. In round $t+3$, agent $a$ moves back to node $u$ and agent $a$ does not move in this round to avoid collision with agent $a$ (a compulsory wait). In round $t+3$, agent $a$ is at $u$ and agent $\hat a$ knows that in the next round agent $a$ will be moving out of this node. In round $t+4$, agent $\hat a$ moves to node $u$ while agent $a$ moves to node $n_u$. This is first time agent $a$ sees agent $\hat a$ moving. Thus, $t+4$ is the event round. In round $t+5$, both agents move back to their respective initial synch nodes. In round $t+5$, each of  the agents knows that the other agent is active and is at its initial synch node. They agree on round $t+6$ ({\it event round}+2) as the red round.

 				%only when it has seen $a_1$ moving at least once and $a_1$ is at node $u$. Let $t$ be the round when agent $a_2$  moves for the first time. The only node where $a_2$ can move from $v$ is node $u$. During round $t$, agents $a_1$ and $a_2$ move to nodes $n_u$ and $u$ respectively. Thus at the end of round $t$, both the agents learn that both of them are active. At the end of round $t$, agent $a_1$ has seen agent $a_2$ moving for the first time and agent $a_1$ is at node $n_u$ which is not its initial synch node.  Since agent $a_1$ has seen agent $a_2$ moving for the first time and $a_2$ was moving out of node $v$, agent $a_1$ knows that $v$ is the initial synch node of $a_2$.  Since both the agents are not at their respective initial synch nodes,  in round  $t+1$, agents  $a_1$ and $a_2$ move back to $u$ and $v$ respectively. At the end of round $t+1$, agent $a_2$ has seen agent $a_1$ moving back to node $u$ even after observing it moving in the previous round. From this agent $a_2$ infers that $u$ is the initial synch node of $a_1$. Thus at the end of round $t+1$, both the agents are aware of the activation of  each other and  they have recognized each other  initial synch nodes. They decide $t+2$ as the {\it red round}. This implies the lemma in this case.  
 			\end{itemize}

 			\item {\bf Case 1.2:  Nodes $u$ and $v$  have at least one common neighbor.} There are two sub-cases.
 			\begin{itemize}
 				\item {\bf Case 1.2.1. None of $B_u$ and $B_v$ is empty.}
 				After the initial wait both agents move in this case. They do oscillating movements to nodes $n_u\in B_u$ and $n_v\in B_v$ respectively. Since $B_u\cap B_v=\emptyset$, the agents do not collide with each other during their movements. In round $t+2$, both agents move and see each other moving for the first time (the event round). In round $t+2$, each of the agents knows that the other agent is active and is currently not at its initial synch node. In round $t+3$, both the agents move back to their respective initial synch nodes. In round $t+3$, both agents know  that the other is active and that it is at its initial synch node. The agents agree on the round $t+4$ ({\it event round}+2) as the {\it red round}.

 				%agent $a$ (respectively $a$) is confirmed about two things: (i) agent $\hat a$ (respectively $a$) knows that agent $a$ (respectively $\hat a$) is active and (ii) it is currently other agent knows that currently   
 				
 				%Suppose $t$ is time when for the first time, they find each other moving. They decide round $t+2$ as the red round. At time $t+1$, if an agent is not at its initial synch node, it moves back to this node. Thus at time $t+2$ both the agents are at their respective initial synch nodes and they have seen each other moving. Agents can identify this case by the fact that they move to the nodes which are common neighbors of their initial synch nodes (where as in case 1.1, agents move to their respective private nodes).
 				
 				\item {\bf Case 1.2.2. One of $B_u$ and $B_v$ is empty.}
 				Without loss of generality, suppose $B_v=\emptyset$. In round $t+2$, agent $a$ moves to node $n_u\in B_u$. This is the first time agent $\hat a$ sees agent $a$ moving. Agent $\hat a$  does not moves in this round (a compulsory wait). In round $t+2$, agent $a$ is at $n_u$ and has not seen agent $\hat a$ moving. In round $t+3$, agent $a$ moves back to node $u$ and agent $\hat a$  moves to node $n_u$. This is the first time both agents move simultaneously. Thus, $t+3$ is the {\it event round}. Since agent $a$ has seen agent $\hat a$ moving for the first time and it is at its initial synch node in round $t+3$, gent $a$ does not move in round $t+4$ (the terminal wait for agent $a$). In round $t+4$, agent $\hat a$ moves back to node $v$. In round $t+4$, each of the agents knows that the other agent is active and is currently at its initial synch node. They agree on the round $t+5$ ({\it event round}+2) as the {\it red} round.    
 				
 				%to $w$ when it has seen $a_1$ moving at least once and $a_1$ is at $w$.    
 				%Suppose $t$ is the time when agent $a_1$ finds $a_2$ moving for the first time. At round $t+1$, agent $a_2$ moves back to $v$ and agent $a_1$ waits at $u$ ({\it terminal wait} for $a_1$). Thus, at the end of round $t+1$, both the agents are at their respective initial synch nodes and they know that both of them are active. They agree on $t+2$ as the {\it red round}. Agent $a_2$ moves only when it is confirmed that $a_1$ is active and is at node $w$. Thus when $a_2$ moves to $w$, agent $a_1$ moves out of $w$ simultaneously. Furthermore, at time $t+1$ only $a_2$ moves back to $v$. Thus they do not collide with each other. In this case, agent $a_1$ does not move after seeing $a_2$ moving while in case 1.2 $a_1$ moves once after seeing $a_2$ moving. This separates these cases completely.  
 			\end{itemize}
 			
 		\end{itemize}
 		\item {\bf Case 2:}  {\it Agents wake up in different rounds.} Without loss of generality, suppose that agent $\hat a$ wakes up later and wakes up in round $t$. It waits in round $t+1$. In round $t+1$, agent $a$ is either moving or waiting for agent $\hat a$ to move first. In this case, agent $a$ knows the  initial synch node of $\hat a$. However, there are cases in which agent $\hat a$ does not know yet the initial synch node of $a$. 
 		
 		\begin{itemize}
 			\item {\bf Case 2.1: Nodes $u$ and $v$ do not have any common neighbor.}
 			\begin{itemize}
 				\item {\bf Case 2.1.1: None of $u$ and $v$ is a leaf node.}  In this case, agent $\hat a$ sees agent $a$ moving between nodes $u$ and $n_u$, in round $t+1$. In round $t+1$, agent $a$ is either at $u$ or at $n_u$.  Agent $\hat a$ does not know the initial synch node of agent $a$. Node $n_u$ is not adjacent to node $v$. Agent $a$ knows the initial synch node of agent $\hat a$. 
 				
 				First suppose that agent $a$ is at node $u$ in round $t+1$. Agent $\hat a$ moves to node $u$ in round $t+2$ while agent $a$ moves to node $n_u$ in this round. This is the first time agents see each other moving, i.e., round $t+2$ is the event round. In round $t+2$ agent $a$ is at $n_u$ which is not the initial synch node of $a$. Thus, in round $t+3$, agent $a$ moves back to $u$ while agent $\hat a$ moves back to $v$. In round $t+3$, agents $a$ and $\hat a$ are  at their respective initial synch nodes and agent $\hat a$ knows  the initial synch node of $a$ in view of the move of agent $a$ in round $t+3$. (Agent $a$ moves even after seeing agent $\hat a$ moving in round $t+2$). They agree on round $t+4$ ({\it event round}+2) as the {\it red round}.

 				Next suppose that agent $a$ is at node $n_u$ in round $t+1$. In round $t+2$ agent $a$ moves to node $u$ and agent $\hat a$ does not move in this round (a compulsory wait to avoid collision). In round $t+2$, agent $a$ is at $u$ and has not seen agent $\hat a$ moving yet. In round $t+3$, agent $a$ moves to node $n_u$ and agent $\hat a$ moves to node $u$. This is the first time both agents see each other moving, i.e., round $t+3$ is the event round. In round $t+3$, both agents are not at their respective initial synch nodes. Thus, in round $t+4$, agents $a$ and $\hat a$ move back to nodes $u$ and $v$ respectively. In round $t+4$, each of the agents knows that the other agent is active and that it is currently at its initial synch node. They agree on round $t+5$ ({\it event round}+2) as the {\it red round}. 
 				
 				\item {\bf Case 2.1.2: One of $u$ and $v$ is a leaf node.} First consider the case when $u$ is a leaf node. In this case,  each of the agents is aware of the initial synch node of the other. Agent $a$ does not move in round $t+1$ (a compulsory wait). In round $t+2$, agent $\hat a$ moves to node $n_v\in A_v$. Nodes $u$ and $n_v$ are not adjacent. In round $t+2$, agent $\hat a$ is at node $n_v$ and has not seen agent $a$ moving yet. Thus, in round $t+3$, agent $\hat a$ moves back to node $v$. In round $t+4$, agent $\hat a$ moves to node $n_v$ and agent $a$ moves to node $v$. This is the first time when agents see each other moving. Thus, round $t+4$ is the event round. In round $t+4$ none of the agents is at its initial synch node. Thus, in round $t+5$ agents $\hat a$ and $a$ move back to nodes $v$ and $u$ respectively. In round $t+5$, each of the agents is at its initial synch node and knows that other agent is active. They agree on round $t+6$ ({\it event round}+2) as the {\it red round}.
 				
 				Now consider the case when $v$ is a leaf node. In this case, agent $\hat a$ finds $a$ moving in round $t+1$ between nodes $u$ and $n_u\in A_u$ and the only node where agent $\hat a$ can move is node $u$. In round $t+1$, agent $a$ is either at $u$ or at $n_u$. The rest of the argument is the same as in Case 2.1.1,
 			\end{itemize}
 			\item {\bf Case 2.2: Nodes $u$ and $v$ have common neighbors.}
 			
 			\begin{itemize}
 				\item   {\bf Case 2.2.1: None of $B_u$ and $B_v$ is empty.} In round $t+1$, agent $a$ moves, and in this round, it is either at $u$ or at $x\in B_u$. Agent $a$ knows the initial synch node of $\hat a$. However, agent $\hat a$ does not know the initial synch node of $a$.  Note that nodes $v$ and $x$ are adjacent. However, nodes $u$ and $v$ may not be adjacent. 
 				
 				First consider the case when nodes $u$ and $v$ are adjacent. If agent $a$ is at node $u$ in round $t+1$, in round $t+2$ agent $\hat a$ moves to node $u$ while agent $a$ moves to node $x$. This is the first time when agents see each other moving simultaneously in the same round i.e., round $t+2$ is the event round. In round $t+2$, agent $a$ has seen agent $\hat a$ moving for the first time, however agent $a$ is not at its initial synch node. Thus, in round $t+3$, agents $a$ and $\hat a$ move back to nodes $u$ and $v$ respectively. During these movements agents do not collide. In round $t+3$, agent $\hat a$ knows the initial synch node of $a$ (since agent $a$ has moved even after watching agent $\hat a$ moving in the previous round), each of the agents knows that the other agent is active and it is at its initial synch node. They agree on round $t+4$ ({\it event round}+2) as the {\it red round}.
 				
 				Now consider the case, when nodes $u$ and $v$ are not adjacent. If agent $a$ is at node $x$ in round $t+1$, agent $\hat a$ moves to node $x$ in round $t+2$ while agent $a$ moves to node $u$ at the same time. This is the first time when agents see each other moving, i.e., round $t+2$ is the event round. In round $t+2$, agent $a$ is at its initial synch node and it has just seen agent $\hat a$ moving for the first time. Thus, in round $t+3$, agent $a$ does not move (the terminal wait) while agent $\hat a$ moves back to node $v$. In round $t+3$, agent $\hat a$ knows the initial synch node of agent $a$ (since agent $a$ has not moved in round $t+3$ after watching agent $\hat a$ moving for the first time in the previous round). During these movements agents do not collide. Thus, in round $t+3$, both agents know that each of them is active and they are at their respective initial synch nodes. They agree on round $t+4$ ({\it event round}+2) as the {\it red round}.
 				
 				If agent $a$ is at node $u$ in round $t+1$, agent $a$ has not seen agent $\hat a$ moving yet. Thus, in round $t+2$, agent $a$ moves from node $u$ to node $x$. Agent $\hat a$ does not move in this round to avoid collision. In round $t+2$, agent $a$ has not seen agent $\hat a$ moving. Thus, in round $t+3$, agents $a$ moves to node $u$ while agent $\hat a$ moves to node $x$ at the same time. This is the first time when both agents see each other moving, i.e., round $t+3$ is the event round. In round $t+3$, agent $a$ is at node $u$ which is its initial synch node and it has seen agent $\hat a$ moving for the first time. Thus, in round $t+4$, agent $a$ does not move (the terminal wait for $a$) while agent $\hat a$ moves back to node $v$. In round $t+4$, each of  the agents knows that the other agent is active and it is at its initial synch node (since in round $t+4$ agent $a$ does not move after watching agent $\hat a$ moving for the first time in round $t+3$, agent $\hat a$ knows the initial synch node of agent $a$ in round $t+4$). They agree on round $t+5$ ({\it event round}+2) as the {\it red round}.
 				
 				\item   {\bf Case 2.2.2: One of $B_u$ and $B_v$ is empty.}  First consider the case when $B_u=\emptyset$. In this case, agent $a$ does not move in round $t+1$. Both agents know the initial synch node of each other. In round $t+2$, agent $\hat a$ moves from $v$ to node $y\in B_v$. Note that nodes $u$ and $y$ are adjacent. This is the first time agent $a$ sees agent $\hat a$ moving and agent $a$ is at its initial synch node. However, agent $a$ has not moved yet and thus round $t+3$ is not a terminal wait round for $a$.  In round $t+3$, agent $\hat a$ moves to node $v$ while agent $a$ moves to node $y$. This is the first time when agents see each other moving simultaneously i.e., round $t+3$ is the  event round. In round $t+3$, agent $\hat a$ is at its initial synch node and it has seen agent $a$ moving for the first time. Thus, in round $t+4$, agent $\hat a$ does not move (the terminal wait for $\hat a$) while agent $a$ moves back to its initial synch node $u$. In round $t+4$, each of the agents knows that other agent is active and currently is at its initial synch node. They agree on round $t+5$ ({\it event round}+2) as the {\it red round}.
 				
 				Now consider the case when $B_v=\emptyset$. In this case agent $\hat a$ sees agent $a$ moving in round $t+1$.  Agent $a$ knows the initial synch node of agent $\hat a$. However, agent $\hat a$ does not know yet the initial synch node of agent $a$. In round $t+1$, agent $a$ is either at $u$ or at a node $x\in B_u$. Nodes $v$ and $x$ are adjacent. There are two possibilities: either nodes $u$ and $v$ are adjacent or not. 
 				
 				First consider the case when nodes $u$ and $v$ are adjacent. If agent $a$ is at node $u$ in round $t+1$, in round $t+2$ agent $a$ moves  to node  $x$ while agent $\hat a$ moves to node $u$ in the same round. This is the first time when agents see each other moving. Thus, round $t+2$ is the event round. In round $t+2$, agent $a$ is at node $x$ and agent $\hat a$ is at node $u$. Thus, in round $t+3$ agents $a$ and $\hat a$ move back to nodes $u$ and $v$, respectively. In round $t+3$, each of the agents knows that the other agent is active and it is at its initial synch node. They decide on round $t+4$ ({\it event round}+2) as the {\it red round}. Next suppose that agent $a$ is at node $x$ in round $t+1$. Since it has not seen  agent $\hat a$ moving in round $t+1$ (agent $\hat a$ does its initial wait in round $t+1$), agent $a$ moves to node $u$ in round $t+2$. In this round, agent $\hat a$ moves to node $x$. This is the first time agents see each other moving, i.e., round $t+2$ is the event round. In round $t+2$, agent $a$ is at its initial synch node and has seen agent $\hat a$ moving for the first time. Thus in round $t+3$, agent $a$ does not move (the terminal wait for $a$) while agent $\hat a$ moves back to node $v$ in the same round. In round $t+3$, each of the agents knows that the other agent is active and currently is at its initial synch node. They agree on round $t+4$ ({\it event round}+2) as the {\it red round}. 
 				
 				Finally consider the case when nodes $u$ and $v$ are not adjacent. In round $t+1$, agent $a$ moves and agent $\hat a$  does not move due to its initial wait. If agent $a$ is at node $x$ in round $t+1$, then in round $t+2$, agent $a$ moves to node $u$ and agent $\hat a$ moves to node $x$. This is the first time when agents see each other moving. Thus, round $t+2$ is the event round. In round $t+2$, agent $a$ is at its initial synch node and has seen the other agent moving for the first time. Thus, in round $t+3$ agent $a$ does not move (the  terminal wait for $a$) while agent $\hat a$ moves back to node $v$. In round $t+3$, agent $\hat a$ is knows the initial synch node of $a$ and both the agents are at their respective initial synch nodes. They agree on round $t+4$ as the {\it red round}. Now, if agent $a$ is at node $u$ in round $t+1$ then, in round $t+2$, agent $a$ moves to node $x$ while agent $\hat a$ does not move in this round (compulsory wait to avoid collision).  In round $t+3$, agents $a$ and $\hat a$ move to nodes $u$ and $x$ respectively. This is the first time when agents see each other moving i.e., round $t+3$ is the event round. Since in round $t+3$, agent $a$ is at its initial synch node and it has seen agent $\hat a$ moving for the first time, agent $a$ does not move in round $t+4$ (the terminal wait for agent $a$). Agent $\hat a$ moves back to node $v$ in this round. In round $t+4$, each of the agents knows the initial synch node of the other agent and both of them are at their respective initial synch nodes. They agree on round $t+5$ ({\it event round}+2) as the {\it red round}.

 			\end{itemize}

 		\end{itemize}
 	\end{itemize}
 	
 	%First observe that the five cases in the description of the subroutine form a disjoint partition of all possible positions of the agents at an approach, the first three cases corresponding to distance 1 between the agents and the last two cases corresponding to distance 2.
 \end{proof}

 We now give a detailed description of subroutine {\tt Progress}. The subroutine starts in the round following the red round. The two agents are at distance at most 2, at nodes $u$ and $v$, which are the nodes at which the current approach started. We consider the subgraph induced by the nodes $u$, $v$ and all their common neighbors. We call it the {\em approach graph} and we say that it is {\em symmetric} if there exists a port preserving automorphism of this graph interchanging $u$ and $v$. If this is not the case, it is possible to elect a leader among the agents, based only on the approach graph. It can be done as follows. Let $L(u)$ (resp. $L(v)$) be the lexicographically ordered list of all edges of the approach graph incident to $u$ (resp. $v$), coded as ordered pairs of port numbers with the port number at $u$ (resp. $v$) first. Since the approach graph is not symmetric, these lists are not identical. Consider the first index $i$ at which they differ. The leader is the agent for which the $i$th edge is lexicographically smaller than the $i$th edge for the other agent. There are four possible cases.
 
 Case 1. Nodes $u$ and $v$ do not have a common neighbor.
 
 Each of the agents executes the next step of their respective exploration walks. Since they do not have a common neighbor, this does not result in a collision.
 
 Case 2. Nodes $u$ and $v$ have a common neighbor, are not adjacent and the approach graph is not symmetric.
 
 If a leader was elected at some previous approach then the other agent is elected as leader in the current approach. Otherwise, agents elect a leader as specified above.
 The leader executes the next step of its exploration walk and the other agent stays idle for a round. Since nodes $u$ and $v$ are not adjacent, there is no collision.
 
 Case 3. Nodes $u$ and $v$ have a common neighbor, are adjacent and the approach graph is not symmetric.
 
 If a leader was elected at some previous approach then the other agent is elected as leader in the current approach. Otherwise, agents elect a leader as specified above.
 If the next step of the exploration walk of the leader is not along the edge joining $u$ and $v$ then the leader executes this next step and the other agent stays idle for a round. Otherwise the leader stays idle for a round. In the next round both agents execute the next step of their respective exploration walks. In both cases there is no collision.
 
 Case 4. Nodes $u$ and $v$ have a common neighbor and the approach graph is symmetric.
 
 This case is the most delicate: agents cannot simply execute the next step of their respective exploration walks because they risk colliding at a common neighbor and they cannot elect a leader  based only on the approach graph because this graph is symmetric. If a leader was elected at some previous approach then the other agent is elected as leader in the current approach.  Otherwise,  agents need to mutually communicate where each of them wants to go next. {\color{black}Let $d$ be the degree of nodes $u$ and $v$ in the original graph} and let $k=\lceil \log d \rceil +1$. Let $s_u$ (resp. $s_v$) be the binary string of length $k$ that is the binary representation of the port number (padded by zeroes in front, if necessary) that the agent at $u$ (resp. at $v$) should take to make the next step of its exploration walk. Each agent communicates this number by walking back and forth on the edge leading to a common neighbor and corresponding to the smallest port number. More precisely, let $x\in \{u,v\}$ and let $w_x$ be the common neighbor of $u$ and $v$ corresponding to this edge for $x$. Due to symmetry, we have $w_u\neq  w_v$. In the next $2k$ rounds $r_1,r'_1,\dots ,r_k,r'_k$, the agent originally located at $x$ behaves as follows. If the $i$th bit of $s_x$ is 1 then the agent goes from $x$ to $w_x$ in round $r_i$ and from $w_x$ to $x$ in round $r'_i$. If the $i$th bit of $s_x$ is 0 then the agent stays idle at $x$ in rounds $r_i$ and $r'_i$. Since the diameter of the approach graph is at most 2, each agent sees the moves of the other agent. After round $r'_k$, both agents learn $s_u$ and $s_v$ and behave as follows in the next round. If $s_u=s_v$ then both agents  execute the next step of their respective exploration walks. Since, by symmetry of the approach graph this is done either on the same edge or on non-incident edges, no collision occurs. If $s_u \neq s_v$ then symmetry can be broken: the leader is the agent for which $s_x$ is smaller. In this case agents behave as in Case 3 or as in Case 2, depending on whether $u$ and $v$ are adjacent or not.
 
 We are now ready to describe our main Algorithm {\tt Explo-without-Collisions}. Each agent executes its exploration walk interrupting it when an approach occurs. At every approach agents execute procedure {\tt Approach resolution} and then continue their respective exploration walks until the next approach. After performing $R(n)$ steps of its exploration walk, an agent stops at some node $v$, possibly only temporarily. At this time the agent has already visited all nodes of the graph. This will be called {\em provisional stop}. Let $w$ be the neighbor of $v$ corresponding to the port 0 at $v$ and let $p$ be the port number at $w$ corresponding to the neighbor $v$. If another approach occurs after the provisional stop of an agent (due to the fact that the other agent has not yet completed its exploration walk) the agent resumes Algorithm {\tt Explo-without-Collisions}, as if the next two port numbers in its exploration walk were 0 and $p$, i.e., it executes procedure {\tt Approach resolution} and then tries to make the next two steps from $v$ to $w$ and back. There is one exception to this rule. Let $\Sigma=R(n)(4\lceil \log n\rceil +8)$. Consider any agent and let $\sigma$ be the round of its provisional stop. 
 If no approach occurred between round $\sigma$ and  $\sigma+\Sigma$  then we define $\rho= \sigma+\Sigma$. Otherwise we define $\rho=\sigma +2\Sigma$.
 The agent knows that no approach will ever occur after round $\rho$ (see the proof of Theorem \ref{final}). It terminates in round $\rho$ and stays idle forever.
 
 \begin{theorem}\label{final}
 	Upon execution of Algorithm {\tt Explo-without-Collisions} both agents have visited all nodes of the graph, were never at the same node in the same round and remain idle forever. The running time of Algorithm {\tt Explo-without-Collisions} is $O(R(n) \log n)$.
 \end{theorem}
 
 \begin{proof}
 	Subroutine {\tt Synchronization} is executed once, at the first approach, and takes constantly many  rounds. Each execution of subroutine {\tt Progress} takes at most $2\lceil \log n\rceil +4$ rounds, as the degree of any node is at most $n$. Hence each agent makes a step of its exploration walk at most every $4\lceil \log n\rceil +8$ rounds. It follows that $\Sigma$ rounds after the beginning of the first execution of subroutine {\tt Progress} both agents have visited all nodes and hence can terminate. On the other hand, if no approach happens by the round $\sigma_A+ \Sigma$, where $\sigma_A$ is the round of the provisional stop of agent $A$ then this agent knows that the other agent $B$ has already visited all nodes of the graph and provisionally stopped. This is because agent $B$ must have woken up before the provisional stop of $A$ (otherwise there would be an approach) and hence the provisional stop of $B$ must have occurred before round $\sigma_A +\Sigma$. To summarize, agent $A$ can safely terminate in round $\rho$ in all cases.
 	
 	The design of procedure {\tt Approach resolution} is such that agents are never at the same node in the same round. In rounds outside of approaches, agents are at distance at least 3, hence there can be no collision in the next round.
 	
 	It remains to estimate the running time of Algorithm {\tt Explo-without-Collisions}. Consider the wakeup round of the later agent as round 0 and consider any agent. Since its provisional stop occurs in round at most $\Sigma + O(1)$, the termination round is at most $3\Sigma + O(1)$. In view of the definition of $\Sigma$ this concludes the proof.
 \end{proof}

 \section{Discussion}
 
 In this section we discuss the two main assumptions of our model and show that none of them can be removed. The first assumption is that agents know an upper bound on the size of the graph, and the second is that agents have vision of radius 2. In order to show the necessity of the first assumption observe that, in the absence of any upper bound on the size of the graph, known to the agents, it is impossible to perform exploration with stop by any team of agents starting simultaneously in some ring in which all ports are numbered 0,1,0,1,... clockwise, regardless of the radius of vision. Indeed, consider a team of $k$ agents with vision of radius $r$, starting simultaneously at distances $r+1$ in rings of size larger than $k(r+1)$.
 For any algorithm that stops after $j$ rounds, no agent can see any other agent, and the subgraph seen by all agents in all rounds is the same. For any such algorithm, no agent will visit all nodes if  the size of the ring is larger than $j+1$.
 
 As for the second assumption, we show that it cannot be weakened. Indeed, suppose that agents have vision of radius 1, and consider 2 agents trying to explore any tree of diameter 2, i.e., any star. Port numbers at the central node $v$ are $0, 1, \dots , d-1$, and port numbers at all leaves are 0. Call the leaf corresponding to port $i$ at the central node, the $i$th leaf. If the agent starts at the $i$th leaf, it sees the two-node graph with port numbers 0 and $i$ at two endpoints of the single edge. Consider any hypothetical collision-free exploration algorithm, where two agents start at leaves 0 and 1. For any $i\in \{0,1\}$, this algorithm must  prescribe a move to the agent starting at the $i$th leaf,  in some round. Suppose that, for the $i$th leaf, it prescribes waiting $x_i\geq 0$ rounds and then moving, unless the agent sees another agent in one of the waiting rounds. Without loss of generality, let $x_0\geq x_1$. Then the adversary starts the agent at leaf 1 $x_0-x_1$ rounds after the agent at leaf 0. Thus both agents collide at the central node after their first move.
 
 \section{Conclusion}
 
 We designed a collision-free exploration algorithm working for two agents in arbitrary graphs of size larger than 2. The natural open problem is to generalize our collision-free exploration to arbitrary teams of agents. This is by no means straightforward. First, it is not even clear if the task is feasible whenever the number of agents is smaller than the number of nodes (as it is for two agents) and if so, what is the minimum vision radius necessary to accomplish it. Second, our exploration strategy for two agents relies on leader election among the agents, if this is possible, and on synchronization. For two agents these are global properties: an order is established on all agents, and all of them (i.e., both in our case) agree on a common round. It is not clear how to achieve a similar feature for many agents. Finally, 
 in our case, agents communicate indirectly by making moves. This is possible because,  in our case,  at least one agent is always adjacent to an empty node, and hence changes of positions can be detected by the other agent. This may not be the case for many agents, as there may be large parts of the graph completely filled with agents, and thus this indirect communication by observing moves may be difficult or impossible.
 %While we showed that arbitrary teams of agents cannot perform collision-free exploration in arbitrary graphs, it could be interesting to investigate the maximum number of agents that can perform collision-free exploration with vision of radius $r$ in $n$-node graphs.
 
 Another issue concerns the efficiency of collision-free exploration, even for two agents. Since any agent can visit only one node in every round, it is impossible to beat the complexity of the best exploration by a single agent. As mentioned before, our solution has a generic feature: given any exploration algorithm by a single agent, it produces a collision-free exploration by two agents, with only logarithmic slowdown. While this slowdown is small, it is natural to ask if it can be removed: does there exist a collision-free exploration algorithm by two agents whose complexity is the same as that of the best exploration by a single agent?  Another variant of the problem is to require that every node be visited by only one agent, instead of both of them. Would this weakening help in getting the same complexity as the best exploration by a single agent? 
 
 \section{Acknowledgment}
 $\textbf{Funding:}$
 Supported in part by Natural Sciences and Engineering Research Council of Canada (NSERC) discovery grant 2018-03899 and by the Research Chair in Distributed Computing of the Universit\'{e} du Qu\'{e}bec en Outaouais, Canada.\\

 % \section{Acknowledgment}
 % 
 % We are grateful to Yoann Dieudonn\'{e} for discussions concerning the topic of this paper.
 
 % \newpage


\begin{thebibliography}{12}
 	
 	
 	\bibitem{AH}
 	S. Albers and M. R. Henzinger,
 	Exploring unknown environments,
 	SIAM Journal on Computing 29 (2000),  1164-1188.  
 	
 	\bibitem{AKLLR}
 	R. Aleliunas, R. Karp, R. Lipton, L. Lovasz C. Rackoff, Random walks, universal traversal sequences, and the complexity of maze problems,
 	Proc. 20th Annual IEEE Symposium on Foundations of Computer Science (FOCS 1979), 218-223.
 	
 	\bibitem{ABRS}
 	Baruch Awerbuch, Margrit Betke, Ronald L. Rivest, Mona Singh:
 	Piecemeal Graph Exploration by a Mobile Robot., Inf. Comput. 152 (1999), 155-172.
 	
 	
 	
 	\bibitem{BBFY}
 	E. Bar-Eli, P. Berman, A. Fiat and R. Yan,
 	On-line navigation in a room,
 	Journal of Algorithms 17 (1994), 319-341.
 	
 	\bibitem{BFRSV}
 	Michael A. Bender, Antonio Fern\'{a}ndez, Dana Ron, Amit Sahai, Salil P. Vadhan:
 	The Power of a Pebble: Exploring and Mapping Directed Graphs., Inf. Comput. 176 (2002), 1-21.
 	
 	\bibitem{BS}
 	M.A. Bender and D. Slonim, The power of team exploration:
 	Two robots can learn unlabeled directed graphs,
 	Proc. 35th Ann. Symp. on Foundations of Computer Science (FOCS 1994),
 	75-85.
 	
 	
 	\bibitem{BRS2}
 	M. Betke, R. Rivest and M. Singh,
 	Piecemeal learning of an unknown environment,
 	Machine Learning 18 (1995), 231-254.
 	
 	\bibitem{BRS}
 	A. Blum, P. Raghavan and B. Schieber,
 	Navigating in unfamiliar geometric terrain,
 	SIAM Journal on Computing 26 (1997), 110-137.
 	
 	%\bibitem{BRT}
 	%A. Borodin, W. Ruzzo, M. Tompa,
 	%Lower bounds on the length of universal traversal sequences,
 	%Journal of Computer and System Sciences 45 (1992), 180-203.
 	
 	
 	
 	\bibitem{CDK}
 	J. Chalopin, S. Das, A. Kosowski,
 	Constructing a map of an anonymous graph: Applications of universal sequences,
 	Proc. 14th International Conference on Principles of Distributed Systems (OPODIS 2010), 119-134.
 	
 	
 	\bibitem{CDGKKP}
 	J. Czyzowicz, D.Dereniowski, L. Gasieniec, R. Klasing, A. Kosowski, Dominik Pajak,
 	Collision-free network exploration. J. Comput. Syst. Sci. 86 (2017), 70-81.
 	
 	
 	
 	
 	
 	\bibitem{DKP}
 	X. Deng, T. Kameda and C. H. Papadimitriou,
 	How to learn an unknown environment I: the rectilinear case,
 	Journal of the ACM 45 (1998), 215-245.
 	
 	
 	\bibitem{DGMPS}
 	A.K. Dhar, B. Gorain, K. Mondal, S. Patra, R.R. Singh,
 	Edge exploration of a graph by mobile agent,
 	Proc. 13th International Conference  on Combinatorial Optimization and Applications (COCOA 2019), 142-154.
 	
 	
 	
 	
 	
 	\bibitem{DFKP}
 	K. Diks, P. Fraigniaud, E. Kranakis, and A. Pelc,
 	Tree exploration with little memory,
 	Journal of Algorithms 51 (2004), 38-63.
 	
 	
 	
 	
 	\bibitem{DKM}
 	S. Dobrev, R. Kralovic, and E. Markou, Online graph exploration with advice,  Proc. 19th International Colloquium on Structural Information and Communication Complexity (SIROCCO 2012), 267-278.
 	
 	
 	
 	
 	
 	\bibitem{DKK}
 	Christian A. Duncan, Stephen G. Kobourov, V. S. Anil Kumar:
 	Optimal constrained graph exploration. ACM Trans. Algorithms 2 (2006), 380-402.
 	
 	
 	
 	
 	
 	
 	
 	\bibitem{FIP2}
 	P. Fraigniaud, D. Ilcinkas, A. Pelc,
 	Tree exploration with advice, Information and Computation 206 (2008), 1276--1287.
 	
 	
 	
 	\bibitem{FI}
 	P. Fraigniaud and D. Ilcinkas.
 	Directed graphs exploration with little memory,
 	Proc. 21st Symposium on Theoretical Aspects of
 	Computer Science (STACS 2004), 246-257.
 	
 	
 	
 	\bibitem{GP}
 	B. Gorain, A. Pelc,
 	Deterministic graph exploration with advice. ACM Trans. Algorithms 15 (2019), 8:1-8:17.
 	
 	
 	\bibitem{Ko}
 	M. Kouck\'{y}, Universal traversal sequences with
 	backtracking, Journal of Computer and System Sciences  65 (2002), 717-726.
 	
 	
 	\bibitem{NMH}
 	Y. Nakaminami, T. Masuzawa, T. Herman, Self-stabilizing agent traversal on tree networks,
 	IEICE Trans. 87-D (12) (20034), 2773-2780.
 	
 	
 	\bibitem{PaPe}
 	P. Panaite and A. Pelc,
 	Exploring unknown undirected graphs,
 	Journal of Algorithms 33 (1999), 281-295.
 	
 	\bibitem{PeTi}
 	A. Pelc, A. Tiane, Efficient grid exploration with a stationary token,
 	International Journal of Foundations of Computer Science 25 (2014), 247-262.
 	
 	\bibitem{RKSI}
 	N. S. V. Rao, S. Kareti, W. Shi and S.S. Iyengar,
 	Robot navigation in unknown terrains: Introductory survey of non-heuristic
 	algorithms,
 	Tech. Report ORNL/TM-12410, Oak Ridge National Laboratory, July 1993.
 	
 	\bibitem{Re}
 	O. Reingold, Undirected connectivity in log-space, Journal of the ACM 55 (2008), 17:1-17:24.
 	
 	
 \end{thebibliography}
\end{document}